\documentclass{article}
\usepackage{spconf,amsmath,graphicx}

\usepackage{setspace}
\usepackage{latexsym}
\usepackage{amsfonts}
\usepackage{curves}
\usepackage{epic}
\usepackage{eepic}
\usepackage{epsfig}
\usepackage{url}
\usepackage{psfrag}
\usepackage{subfigure}
\usepackage{times}
\usepackage{color}
\usepackage{helvet}
\usepackage{caption}
\usepackage{courier}
\usepackage{amsmath, amsthm, amssymb}
\usepackage{graphicx}

\usepackage{cite}
\usepackage{subfigure}

\newcommand{\dprime}{{\prime\prime}}

\newtheorem{proposition}{Proposition}

\newtheorem{algorithm}{Algorithm}
\newtheorem{definition}{Definition}

\newtheorem{lemma}{Lemma}

\newtheorem{fact}{Fact}

\newcommand{\setn}{\mathcal{N}}
\newcommand{\setm}{\mathcal{M}}

\newcommand{\signal}[1]{{\boldsymbol{#1}}}

\newcommand{\Natural}{{\mathbb N}}

\newcommand{\real}{{\mathbb R}}

\newcommand{\refeq}[1]{(\ref{#1})}

\title{Power Estimation in {LTE} systems with the General\\ Framework of Standard Interference Mappings}
\name{R.~L.~G.~Cavalcante, E. Pollakis, and S. Sta\'nczak}
\address{Fraunhofer Heinrich Hertz Institute, Einsteinufer 37, 10587 Berlin, Germany\\
email: \{renato.cavalcante, emmanuel.pollakis,slawomir.stanczak\}@hhi.fraunhofer.de }

\begin{document}

\maketitle
\begin{abstract}
We devise novel techniques to obtain the downlink power inducing a given load in long-term evolution (LTE) systems, where we define load as the fraction of resource blocks in the time-frequency grid being requested by users from a given base station. These techniques are particularly important because previous studies have proved that the data rate requirement of users can be satisfied with lower transmit energy if we allow the load to increase. Those studies have also shown that obtaining the power assignment inducing a desired load profile can be posed as a fixed point problem involving standard interference mappings, but so far the mappings have not been obtained explicitly. One of our main contributions in this study is to close this gap. We derive an interference mapping having as its fixed point the power assignment inducing a desired load, assuming that such an assignment exists. Having this mapping in closed form, we simplify the proof of the aforementioned known results, and we also devise novel iterative algorithms for power computation that have many numerical advantages over previous methods.
\end{abstract}

\section{introduction}
In long-term evolution (LTE) systems, user data is transmitted in the time-frequency domain in basic units called resource blocks. In the current LTE standard, users served by the same base station do not cause interference to each other because base stations assign different resource blocks to their own users. However, owing to the scarcity of the wireless spectrum, the time-frequency grid is reused by the base stations, so users connected to different base stations can interfere with each other. This type of interference can severely limit the downlink data rates, so good interference coupling models are required to determine whether a given data rate requirement can be supported by the network \cite{ho2014,renato14SPM,Majewski2010,siomina12,feh2013,MajewskiTurkeHuangEtAl2007Analytical}, in which case the network is said to be feasible. This information is used for various network optimization tasks, such as energy optimization \cite{renato14SPM,pollakis12}.

To date, the feasibility of LTE-like networks are typically demonstrated by computing the fraction of resource blocks that each base station requires to support the traffic demand \cite{Majewski2010,siomina12}, where all network parameters, such as the downlink transmit power per resource block, are assumed to be fixed. This fraction is commonly referred to as load, and, as shown in \cite{feh2013,renato14SPM}, many previous results on load analysis in LTE-like systems can be unified and generalized by using the framework of interference calculus \cite{yates95,martin11,stanczak2009}. \par 

More recently, the study in \cite{ho2014} has highlighted the importance of a problem strongly related to that of load characterization in networks; namely, that of computing the downlink transmit power inducing a given load profile. By using standard load coupling models for LTE-like systems, the authors of that study prove that the users' rate requirements can be satisfied with lower transmit power if the load at each base station is allowed to increase. Furthermore, they also show that obtaining the power assignment of base stations inducing a given load profile can be posed as a fixed point problem involving interference mappings, and an algorithm to compute the fixed point is developed. However, the interference mappings have not been obtained in closed form, and the algorithm for power computation requires two nested iterative methods that only converge asymptotically. One of the shortcomings of the algorithm is that the inner iterative method may require many iterations to obtain good numerical accuracy, which can be costly in terms of computational time and effort. 

In this study, we derive, in closed-form, an interference mapping having as its fixed point the power allocation inducing a given load profile. By doing so, we are able to simplify the proof of some previous results, and we are also able to derive novel algorithms for power computation that do not require nested iterative techniques such as that in \cite{ho2014}. In particular, if we are given the task to recompute power assignments to increase load (e.g., to decrease the transmit energy as discussed above), we can derive simple algorithms that can give information about the precision of the power estimates at each iteration.

\section{Preliminaries}
In this section we reproduce standard results that are extensively used in this study, and we note that much of the material here can has been taken directly from \cite{renato14SPM}. Hereafter, inequalities involving vectors should be understood as element-wise inequalities. Furthermore, $\real_+$ denotes the set of non-negative real numbers, and $\real_{++}$ is the set of strictly positive numbers. Unless otherwise stated, the $i$th component of a vector $\signal{x}$ (vectors are always written in bold typeface) is denoted by $x_i$.

\begin{definition}
\label{definition.inter_func}
 (Interference functions and mappings \cite{yates95,martin11,stanczak2009}) A function $I: \real_+^M \to \real_{++}$ is said to be a standard interference function if the following properties hold:\par 
\begin{enumerate}
 \item ({\it Scalability}) $\alpha {I}(\signal{x})>I(\alpha\signal{x})$ for all $\signal{x}\in\real^M_+$ and all $\alpha>1$. \par
 \item ({\it Monotonicity}) ${I}(\signal{x}_1)\ge I(\signal{x}_2)$ if $\signal{x}_1\ge \signal{x}_2$. \par 
\end{enumerate}
Given $M$ standard interference functions $I_i:\real^M_+\to\real_{++}$, $i=1,\ldots,M$, we call the mapping $\mathcal{J}:\real^M_+\to\real_{++}^M$ with $\mathcal{J}(\signal{x}):=[I_1(\signal{x}),\ldots, I_M(\signal{x})]^T$ a ``standard interference mapping'' or simply ``interference mapping.''
\end{definition}

In the later sections, we estimate load and power of networks by computing fixed points of standard interference mappings, and the fact shown below is useful for this purpose.

\begin{fact} \label{fact.inter_map} (Properties of interference mappings \cite{yates95}) 
\begin{enumerate} 
\item If a standard interference mapping $\mathcal{J}:\real^M_{+}\to\real^M_{++}$ has a fixed point $\signal{x}\in\mathrm{Fix}(\mathcal{J}):=\{\signal{x}\in\real^M_{++}~|~\mathcal{J}(\signal{x})=\signal{x}\}$, then the fixed point is unique.
\item A standard interference mapping $\mathcal{J}:\real^M_+\to\real^M_{++}$ has a fixed point if and only if there exists $\signal{x}^\prime\in\real^M$ satisfying $\mathcal{J}(\signal{x}^\prime)\le\signal{x}^\prime$.
\item If a standard interference mapping $\mathcal{J}:\real^M_+\to\real^M_{++}$ has a fixed point, then it is the limit of the sequence $\{\signal{x}_n\}$ generated by $\signal{x}_{n+1}=\mathcal{J}(\signal{x}_n)$, where $\signal{x}_1\in\real^M_{+}$ is arbitrary. In particular, if $\signal{x}_1=\signal{0}$, then the sequence is monotonously increasing (in each component). In contrast, if $\signal{x}_1$ satisfies $\mathcal{J}(\signal{x}_1)\le \signal{x}_1$, then the sequence is monotonously decreasing (in each component).
\end{enumerate}
\end{fact}

In many cases, identifying interference functions by using the following results can be easier than by verifying the properties in Definition~\ref{definition.inter_func}. 

\begin{fact}
\label{fact.pos_conc}
 Concave functions $I:\real^M_{+}\to\real_{++}$ are standard interference functions \cite{renato14SPM}.
\end{fact}

In turn, to prove that a given function is concave, we can start with a simple function that is known to be concave and reconstruct the function under consideration by using operations that preserve concavity.
\begin{fact} (Selected concavity preserving operations)
\label{fact.preserve}
Let $f:\real^M\to\real$ be a concave function. We can use the following operations to obtain new concave functions \cite[Sect.~2.3]{boyd} \cite[Sect.~8.2]{baus11}:
\begin{enumerate}
 \item Let $\signal{A}\in\real^{M\times N}$ and $\signal{b}\in\real^M$ be arbitrary, and define $g:\real^N\to\real^M:\signal{x}\mapsto \signal{Ax}+\signal{b}$. Then $f\circ g:\real^N\to\real:\signal{x}\mapsto f(g(\signal{x}))$ is concave.
 \item (Perspective) The perspective function $g:\real^M\times \real_{++}\to\real:(\signal{x},P)\mapsto P\cdot~f(1/P~\signal{x})$ (associated to $f$) is concave. 
 \item (Dimension reduction) Fixing arguments of concave functions preserve concavity. For example, the function $g:\real^{N-1}\to\real:\signal{x}\mapsto f([\signal{x}^T~ 1]^T)$, which is obtained by fixing the last element argument of the function $f$ to one, is a concave function.
 \item (Scalar multiplication and addition) Concave functions are preserved under addition and multiplication by strictly positive constants.
\end{enumerate}

\end{fact}

\section{Power estimation in wireless networks}

\subsection{System model and problem statement}
In this study, we use a LTE interference model that has been studied for many years in the literature \cite{ho2014,renato14SPM,Majewski2010,siomina12,feh2013,MajewskiTurkeHuangEtAl2007Analytical}. In more detail, we denote by $\setm=\{1,\ldots,M\}$ the set of $M$ LTE base stations serving at least one user, and by $\setn=\{1,\ldots,N\}$ the set of $N$ users requesting service from base stations. Each user $j\in\setn$ requests a fixed data rate $d_j\in\real_{++}$, and $\setn_i$ is the set of users connected to the $i$th base station. We assume that pathloss between user $j\in\setn$ and base station $i\in\setm$ is denoted by $g_{i,j}\in\real_{++}$, and note that the assumption $g_{i,j}\neq 0$ is used for brevity. The effective bandwidth of each resource block is denoted by $B\in\real_{++}$, and there are $K$ resource blocks in the system. Each base station $i\in\setm$ transmits with fixed power $p_i\in\real_{++}$ per resource block. If user $j$ is served by base station $i$, the reliable downlink data rate per resource block is approximated by:
\begin{align}
 \omega_{i,j}(\signal{\nu},\signal{p})=B\log_2\left(1+\dfrac{p_i g_{i,j}}{\sum_{k\in\setm\backslash\{i\}}\nu_k p_k g_{k,j}+\sigma^2}\right),
\end{align}
where $\sigma^2$ is the noise power per resource block, $\signal{p}=[p_1,\ldots,p_M]^T$ is the downlink power vector per resource block, and $\signal{\nu}=[\nu_1,\ldots,\nu_M]^T$ is the load at the base stations. Here, the load $\nu_i$ at the $i$th base station is the fraction of resource blocks being used at base station $i$ for data transmission. For a fixed power assignment $\signal{p}\in\real_{++}^M$, the load vector can be obtained by solving the following system of nonlinear equations:
\begin{align}
\label{eq.nonlinear}
 \nu_i = \sum_{j\in\setn_i} \dfrac{d_j}{K\omega_{i,j}(\signal{\nu},\signal{p})}, ~i\in\setm,
\end{align}
or, equivalently, by computing the fixed point of the mapping 
\begin{align}
\label{eq.load_mapping}
\mathcal{J}_\signal{p}:\real^M_+\to\real^M_{++}:\signal{\nu}\mapsto[I_{\signal{p},1}(\signal{\nu}),\ldots,I_{\signal{p},M}(\signal{\nu})]^T,
\end{align}
where $I_{\signal{p},i}(\signal{\nu}):=\sum_{j\in\setn_i} \dfrac{d_j}{K\omega_{i,j}(\signal{\nu},\signal{p})}$. The mapping $\mathcal{J}_\signal{p}$ is a standard interference mapping \cite{renato14SPM,feh2013,ho2014}, so, by Fact.~\ref{fact.inter_map}, the fixed point, if it exists, can be obtained, for example, with the standard iterative algorithm $\signal{\nu}_{n+1}=\mathcal{J}_\signal{p}(\signal{\nu}_n)$ with $\signal{\nu}_1\in\real_+^M$ arbitrary. Note that, if the fixed point $\signal{\nu}^\star\in\mathrm{Fix}(\mathcal{J}_\signal{p})$ exists, the total transmit power of base station $i\in\setm$ is given by $K\nu_i^\star p_i$. \par 

Recently, the study in \cite{ho2014} has highlighted the importance of the reverse problem; namely, that of solving the nonlinear system in \refeq{eq.nonlinear} for the power allocation $\signal{p}$ with the load $\signal{\nu}$ fixed. In particular, energy efficiency power allocations can be obtained by solving the reverse problem, which is the problem we study here.

\subsection{Interference functions for the computation of the power vector}
To solve the nonlinear system in \refeq{eq.nonlinear} for the power vector $\signal{p}\in\real_{++}^M$, with the load $\signal{\nu}\in\real_{++}^M$ and all other parameters of the model remaining fixed, we start by multiplying both sides of \refeq{eq.nonlinear} by $p_i/\nu_i>0$:
\begin{align}
\label{eq.nonlinear2}
 p_i = \tilde{\mathcal{P}}_{\signal{\nu},i}(\signal{p}),\quad i\in\setm,
\end{align}
where 
\begin{align}
\label{eq.func_power}
\tilde{\mathcal{P}}_{\signal{\nu},i}:\real^M_{++}\to\real_{++}:\signal{p}\mapsto\mathcal\mathcal{}\dfrac{p_i}{\nu_i} \sum_{j\in\setn_i} \dfrac{d_j}{K\omega_{i,j}(\signal{\nu},\signal{p})},
\end{align}
 Note that, by construction, $\signal{p}\in\real^M_{++}$ solves the system in \refeq{eq.nonlinear} if and only if it also solves the nonlinear system in \refeq{eq.nonlinear2}. In the remaining of this subsection, we show that if these systems have a solution, the solution is the fixed point of a standard interference mapping that we obtain in closed form.  We start with the following simple result.
 
\begin{proposition}
\label{lemma.concavity}
 The function $\tilde{\mathcal{P}}_{\signal{\nu},i}:\real^M_{++}\to\real_{++}$ defined in \refeq{eq.func_power} is concave for every $i\in\setm$.
\end{proposition}
\begin{proof}
 Let $\signal{p}_{-i}\in\real^{M-1}_{++}$ be a power vector obtained by excluding the $i$th component of the power vector $\signal{p}$, where $i$ is arbitrary. By noticing that the function $f^{(1)}:\real\to\real:x\mapsto 1/\log_2(1+1/x)$ is concave, by Fact.~\ref{fact.preserve}.1, we readily verify that the function
 
 \begin{align*} 
f_{i,j}^{(2)}:\real^{M}_{++}&\to\real_{++} \\ 
\left[\begin{matrix}\signal{p}_{-i}\\y\end{matrix}\right] &\mapsto \dfrac{d_j}{BK \log_2\left(1+\dfrac{g_{i,j}}{\sum_{k\in\setm\backslash\{i\}}\nu_k p_k g_{k,j}+y}\right)}
\end{align*}
is concave for arbitrary $i\in\setm$ and $j\in\setn$. Therefore, by Fact.~\ref{fact.preserve}.2, the function $f_{i,j}^{(3)}:\real_{++}^{M+1}\to\real_{++}:\left[\begin{matrix}\signal{p}\\ y\end{matrix}\right]\mapsto p_i f_{i,j}^{(2)}\left(\dfrac{1}{p_i}\left[\begin{matrix}\signal{p}_{-i} \\ y\end{matrix}\right]\right)$ is concave. We can now fix $y=\sigma^2$ and apply Fact.~\ref{fact.preserve}.3 to $f_{i,j}^{(3)}$ to show that $f_{i,j}^{(4)}(\signal{p}):=p_i d_j/(K\omega_{i,j}(\signal{\nu},\signal{p}))$ is concave. Concavity of $\mathcal{P}_{\signal{\nu},i}$ now follows from this last result and Fact.~\ref{fact.preserve}.4.
\end{proof}

We now continuously extend $\tilde{\mathcal{P}}_{\signal{\nu},i}$ to the closure of its domain (the proof of the next lemma will be shown elsewhere).

\begin{lemma}
\label{lemma.extension}
 For every $i\in\setm$, the concave function $\tilde{\mathcal{P}}_{\signal{\nu},i}:\real^M_{++}\to\real_{++}$ can be continuously extended to the domain $\real_{+}^M$. This extension, denoted by $\mathcal{P}_{\signal{\nu},i}$, which is also a concave function, is given by
 \begin{align}
 \label{eq.extension}
  \mathcal{P}_{\signal{\nu},i}(\signal{p})=\begin{cases}
       \dfrac{p_i}{\nu_i} \sum_{j\in\setn_i} \dfrac{d_j}{K\omega_{i,j}(\signal{\nu},\signal{p})},
 \quad \mathrm{if}~~ p_i\ne 0  \\ 
 \sum_{j\in\setn_i} \dfrac{d_j\ln 2}{KBg_{i,j}\nu_i}\left(\sum_{k\in\setm\backslash\{i\}}\nu_k p_k g_{k,j}+\sigma^2\right), \\ \qquad\qquad \mathrm{otherwise,}
      \end{cases}
 \end{align}
 and its codomain is $\real_{++}$.
\end{lemma}
 
The next proposition shows that the solution of the system in \refeq{eq.nonlinear2} (or, equivalently, \refeq{eq.nonlinear} with $\signal{p}$ being the variable to be determined) is the fixed point of a standard interference mapping.
\begin{proposition}
\label{proposition.main_result}
 Define the mapping $\mathcal{P}_\signal{\nu}:\real^M_+\to\real^M_{++}$ by $\mathcal{P}_\signal{\nu}(\signal{p}):=[\mathcal{P}_{\signal{\nu},1}(\signal{p}),\ldots,\mathcal{P}_{\signal{\nu},M}(\signal{p})]^T$, where $\mathcal{P}_{\signal{\nu},i}$ is given in \refeq{eq.extension}. Then $\mathcal{P}_\signal{\nu}$ is a standard interference mapping, and its fixed point, if it exists, is unique, and it coincides with the solution of the nonlinear system in \refeq{eq.nonlinear2}.
\end{proposition}
\begin{proof}
We have already proved in Lemma~\ref{lemma.extension} that $\mathcal{P}_{\signal{\nu},i}$ is a positive concave function for every $i\in\setm$. As a result, we can apply Fact.~\ref{fact.pos_conc} to conclude that the mapping $\mathcal{P}_\signal{\nu}$ is a standard interference mapping. By Fact.~\ref{fact.inter_map}.1, the fixed point $\signal{p}^\star\in\mathrm{Fix}(\mathcal{P}_\signal{\nu})$, if it exists, is unique (and strictly positive). These facts imply the equivalence between the solution of the nonlinear system in \refeq{eq.nonlinear2} and the fixed point $\signal{p}^\star$ of $\mathcal{P}_\signal{\nu}$.
\end{proof}

A practical consequence of the above proposition is that the power assignment $\signal{p}$ inducing a given load $\signal{\nu}$ (if it exists) is the limit of the sequence $\{\signal{p}_n\}$ generated by $\signal{p}_{n+1}=\mathcal{P}_\signal{\nu}(\signal{p}_n)$, where $\signal{p}_1\in\real_{+}^M$ is arbitrary. Note that this simple iterative scheme eliminates the need for the bisection technique required by the scheme in \cite{ho2014}. 

For the reasons shown below, we are often interested in increasing the load of the current network configuration by changing the power $\signal{p}$, and, for this task, we can devise an iterative algorithm that also provides information about the precision obtained at each iteration. 

\subsection{Iterative algorithms for power planning}
\label{sect.power_planning}

Suppose that a power assignment $\signal{p}^\prime$ induces a load $\signal{\nu}^\prime$. Now, assume that we increase the load from $\signal{\nu}^\prime$ to $\signal{\nu}^\dprime\ge \signal{\nu}^\prime$ (with $\signal{\nu}^\prime\neq\signal{\nu}^\dprime$) by changing the power from $\signal{p}^\prime$ to $\signal{p}^\dprime$ while keeping all other parameters of the model fixed. In Proposition~\ref{proposition.load} below, we prove that $\signal{p}^\dprime<\signal{p}^\prime$ and that $\nu_i^\dprime p_i^\dprime < \nu_i^\prime p_i^\prime$ for every $i\in\setm$ . In particular, this last inequality shows that the users' data rate requirements can be satisfied with lower transmit power if we allow the load to increase. We emphasize that this conclusion is not our original contribution because it has been originally obtained in \cite{ho2014}. However, our proof is new because it uses the interference mapping $\mathcal{P}_\signal{\nu}$ obtained in Proposition~\ref{proposition.main_result}. The results in Proposition~\ref{proposition.load} are also used to derive a novel algorithm for power computation, and the proof of this proposition requires the following lemma.

\begin{lemma}
\label{lemma.log}
 The function $f:\real_{++}\to\real_{++}:x\mapsto x \ln(1+1/x)$ is strictly increasing; i.e., $y,x\in\real_{++}$ with $y>x$ implies $f(y)>f(x)$.
\end{lemma}
\begin{proof}
 First recall that $\dfrac{y}{y+1}<\ln(1+y)$ for every $y>0$ \cite{love80}. Now, for $y=1/x\in\real_{++}$, we deduce: 
$0 < \ln\left(1+\dfrac{1}{x}\right)-\dfrac{1}{1+x} = f^\prime(x)$ for every $x\in\real_{++}$,
which implies the desired result.
\end{proof}

\begin{proposition}
\label{proposition.load}
 Let $\signal{\nu}^\prime\in\real_{++}^M$ be the load corresponding to the power assignment $\signal{p}^\prime\in\real_{++}^M$; i.e., $\signal{\nu}^\prime\in\mathrm{Fix}(\mathcal{J}_{\signal{p}^\prime})$, or, equivalently, $\signal{p}^\prime\in\mathrm{Fix}(\mathcal{P}_{\signal{\nu}^\prime})$. Choose an arbitrary vector satisfying $\signal{\nu}^\dprime\ge\signal{\nu}^\prime$ and $\signal{\nu}^\prime\neq\signal{\nu}^\dprime$, and define $\alpha_i=\nu^\dprime_i/\nu^\prime_i\ge 1$ for $i\in\setm$.
Then the interference mapping $\mathcal{P}_{\signal{\nu}^\dprime}:\real_{+}^M\to\real_{++}^M$ has a uniquely existing fixed point $\signal{p}^\dprime\in \real^M_{++}$. Furthermore, we have $\signal{0}<\signal{p}^\dprime<\signal{p}<\signal{p}^\prime$ and $\nu_i^\dprime p_i^\dprime < \nu_i^\prime p_i^\prime$ for every $i\in\setm$, where the $i$th element of the vector $\signal{p}=[p_1,\ldots,p_M]^T$ is given by $p_i:=p_i^\prime/\alpha_i$. Moreover, the sequence $\{\mathcal{P}^n_{\signal{\nu}^\dprime}(\signal{p})\}_{n\in\Natural}$, which converges to $\signal{p}^\dprime$, is monotonously decreasing.
\end{proposition}
\begin{proof}
 By definition, $p_i\nu_i^\dprime = p_i^\prime\nu_i^\prime$ for every $i\in\setm$. As a result, by Lemma~\ref{lemma.log} and $\signal{p}^\prime\in\mathrm{Fix}(\mathcal{P}_{\signal{\nu}^\prime})$, we deduce
 \begin{multline}
 \label{eq.ineq}
\mathcal{P}_{\signal{\nu}^\dprime,i}(\signal{p}) =\\ \dfrac{p_i^\prime}{\alpha_i\nu_i^\prime} \sum_{j\in\setn_i} \dfrac{d_j}{\alpha_i BK\log_2\left(1+\dfrac{p_i^\prime g_{i,j}}{\alpha_i\left(\sum_{k\in\setm\backslash\{i\}}\nu_k^\prime p_k^\prime g_{k,j}+\sigma^2\right)}\right)}\\ 
\le \dfrac{p_i^\prime}{\alpha_i\nu_i^\prime}\sum_{j\in\setn_i}\dfrac{d_j}{K\omega_{i,j}(\signal{\nu}^\prime,\signal{p}^\prime)} = \dfrac{1}{\alpha_i} \mathcal{P}_{\signal{\nu}^\prime,i}(\signal{p}^\prime)=\dfrac{p_i^\prime}{\alpha_i}=p_i,
\end{multline}
and the inequality is strict if and only if $i\in\mathcal{I}:=\{k\in\setm~|~\alpha_k>1\}\neq\emptyset$. Therefore, $\mathcal{P}_{\signal{\nu}^\dprime}(\signal{p})\le \signal{p}$, which is already enough to show by Fact.~\ref{fact.inter_map} that the fixed point $\signal{p}^\dprime$ of the mapping $\mathcal{P}_{\signal{\nu}^\dprime}$ exists, it is unique, and it satisfies $\signal{p}^\dprime \le \mathcal{P}^n_{\signal{\nu}^\dprime}(\signal{p})\le \signal{p}$ for every $n\in\Natural$. This last inequality and Fact.~\ref{fact.inter_map} also show that the sequence $\{\mathcal{P}^n_{\signal{\nu}^\dprime}(\signal{p})\}_{n\in\Natural}$ is monotonously decreasing (and converges to $\signal{p}^\dprime\in\mathrm{Fix}(\mathcal{P}_{\signal{\nu}^\dprime})$). From \refeq{eq.ineq} and the assumption that $g_{i,j}>0$ for every $i\in\setm$ and $j\in\setm$,\footnote{If we replace this assumption by the weaker assumption that only the pathlosses between users and their serving base stations are not zero, then the next strict inequalities should be replaced by their corresponding nonstrict inequalities.} we observe that $\mathcal{P}_{\signal{\nu}^\dprime,i}(\signal{p})\nu_i^\dprime<p_i\nu_i^\dprime=p_i^\prime\nu_i^\prime$  for every $i\in\mathcal{I}$ (for $i\notin\mathcal{I}$, we have $\mathcal{P}_{\signal{\nu}^\dprime,i}(\signal{p})=p_i$). We can now verify that $\mathcal{P}^2_{\signal{\nu}^\dprime}(\signal{p})<\signal{p}$, which, by Fact.~\ref{fact.inter_map}, shows that $\signal{p}^\dprime < \signal{p}$, and we conclude that $p_i^\dprime\nu^\dprime_i<p_i\nu_i^\dprime = p_i^\prime\nu_i^\prime$ for every $i\in\setm$.
\end{proof}

We now derive a simple algorithm based on \cite[Remark 1]{renato14SPM}. The objective of the algorithm is to compute new power assignments to increase the load of a given network configuration (as proved above, and also in \cite{ho2014}, by doing so we decrease the transmit power). In more detail, let $\signal{p}^\prime\in\mathrm{Fix}(\mathcal{P}_{\signal{\nu}^\prime})$ and $\signal{\nu}^\prime\in\mathrm{Fix}(\mathcal{J}_{\signal{p}^\prime})$ be the power and load for the current network configuration, respectively. To compute a new power assignment $\signal{p}^\dprime$ inducing a load $\signal{\nu}^\dprime\ge \signal{\nu}^\prime$, while keeping all other parameters constant (e.g., the users' data rates), we can proceed as follows. With the standard iteration $\signal{p}_{n+1}=\mathcal{P}_{\signal{\nu}^\dprime}(\signal{p}_n)$, construct in parallel two sequences $\{\overline{\signal{p}}_n\}$ and $\{\underline{\signal{p}}_n\}$ where $\underline{\signal{p}}_1:=\signal{0}$, $\overline{\signal{p}}_1 := \signal{p}$, and $\signal{p}$ is the vector defined in Proposition~\ref{proposition.load}. Fact.~\ref{fact.inter_map} and Proposition~\ref{proposition.load} show that the sequences $\{\overline{\signal{p}}_n\}$ and $\{\underline{\signal{p}}_n\}$ are monotonously decreasing and increasing, respectively, and both sequences converge to $\signal{p}^\dprime\in \mathrm{Fix}(\mathcal{P}_{\signal{\nu}^\dprime})\neq\emptyset$. As a result, $\underline{\signal{p}}_n\le\signal{p}^\dprime\le \overline{\signal{p}}_n$ for every $n\in\Natural$, and the monotonously decreasing sequence $\{{\epsilon}_n:=\|\underline{\signal{p}}_n-\overline{\signal{p}}_n\|_\infty\}$ provide us with information about the numerical precision obtained at each iteration $n$ because we have both $\|\underline{\signal{p}}_n-\signal{p}^\dprime\|_\infty\le {\epsilon}_n$ and $\|\overline{\signal{p}}_n-\signal{p}^\dprime\|_\infty\le {\epsilon}_n$. These facts suggest the following algorithm.

\begin{algorithm}
\label{algorithm.power_increase}
 {\bf Input:} Current load $\signal{\nu}^\prime$, current power assignment $\signal{p}^\prime$, desired load $\signal{\nu}^\dprime\ge\signal{\nu}^\prime$, maximum number of iterations $m$, vector $\signal{p}$ defined in Proposition~\ref{proposition.load}, and desired numerical precision $\epsilon>0$ of the power assignment $\signal{p}^\dprime$ inducing the load $\signal{\nu}^\dprime$. \\
 {\bf Output:} Power assignment $\widetilde{\signal{p}}$ and numerical precision $\widetilde{\epsilon}$ satisfying $\|\widetilde{\signal{p}}-\signal{p}^\dprime\|_\infty\le\widetilde{\epsilon}$. \\
 {\bf Initialization:} $\underline{\signal{p}}\leftarrow \signal{0}$, $\overline{\signal{p}}\leftarrow \signal{p}$, $n\leftarrow 0$, $\widetilde{\epsilon}=\|\signal{p}\|_\infty$. \\
 {\bf Algorithm:} \\While $\widetilde{\epsilon} > \epsilon$ and $n \le m$ do:\par 
 $\underline{\signal{p}} \leftarrow {\mathcal{P}_\signal{\nu^\dprime}}(\underline{\signal{p}})$, $\overline{\signal{p}} \leftarrow {\mathcal{P}_\signal{\nu^\dprime}}(\overline{\signal{p}})$, $\widetilde{\epsilon} \leftarrow \|\underline{\signal{p}}-\overline{\signal{p}}\|_\infty$, $n\leftarrow n+1$ \\
 Return $\widetilde{\signal{p}} \leftarrow \overline{\signal{p}}$ and $\widetilde{\epsilon}$.
\end{algorithm}
We note that, by Fact.~\ref{fact.inter_map}, the above algorithm terminates after a finite number of iterations even if we set $m=\infty$, in which case $\widetilde{\epsilon}\le\epsilon$ upon termination.

\section{Conclusion}
We have derived a standard interference mapping that has as its fixed point the power allocation inducing a given load in LTE-like systems, and we highlighted some of the benefits of having the mapping in closed form. For example, we showed that well-known techniques to compute fixed points become readily available, and these iterative techniques are remarkably simpler than previous methods that, for example, require nested iterative approaches. In particular, one the proposed iterative techniques is able to give accurate information about the precision of the power assignment vector obtained at each iteration. We also showed that knowledge of the mapping can be used to simplify the proof of results obtained in recent studies (e.g., the proof that increasing load by changing the power allocations reduces the transmit energy). \par

{\bf\small\bf Acknowledgements}{\small: This work has been performed in the framework of the FP7 project ICT-317669 METIS, which is partly funded by the European Union. The authors would like to acknowledge the contributions of their colleagues in METIS, although the views expressed are those of the authors and do not necessarily represent the project. }

\bibliographystyle{IEEEtran}
\bibliography{IEEEabrv,references}
\end{document}